\documentclass[conference,a4paper]{IEEEtran}
\usepackage{balance}  
\normalsize
\usepackage{color}
\usepackage{graphicx}
\usepackage{amsfonts}
\usepackage[cmex10]{amsmath}
\usepackage{cite}
\usepackage{algorithmic}
\usepackage[center]{caption}
\usepackage{epsfig}
\usepackage{latexsym}
\usepackage{epstopdf}
\usepackage{verbatim}
\usepackage{amsthm}
\usepackage{bbm}
\usepackage{chngcntr}
\usepackage{mathtools}
\graphicspath{{./figures/}}
\theoremstyle{definition}

\newtheorem{theorem}{Theorem}

\newtheorem{remark}{Remark}
\newtheorem{lemma}{Lemma}
\newcounter{numrel}
\counterwithin*{numrel}{section}

\newcommand{\remove}[1]{}
 \usepackage{amssymb}
 \allowdisplaybreaks
\begin{document}
\title{Consistency Analysis of Replication-Based Probabilistic Key-Value Stores
} 

\author{Ramy E. Ali }

\maketitle
\begin{abstract}
Partial quorum systems are widely used in distributed key-value stores due to their latency benefits at the expense of providing weaker consistency guarantees. The probabilistically bounded staleness framework (PBS) studied the latency-consistency trade-off of Dynamo-style partial quorum systems through Monte Carlo event-based simulations. In this paper, we study the latency-consistency trade-off for such systems analytically and derive a closed-form expression for the inconsistency probability. Our approach allows fine-tuning of latency and consistency guarantees in key-value stores, which is intractable using Monte Carlo event-based simulations.  
\end{abstract}
\makeatletter{\renewcommand*{\@makefnmark}{}
	\footnotetext{\hrule \vspace{0.05in}Ramy E. Ali (email: reali@usc.edu) is with Ming Hsieh Department of Electrical Engineering, University of Southern California, Los Angeles, CA 90089. This work was done while the author was with the School of Electrical Engineering and Computer Science,  Pennsylvania State University, University	Park, PA 16802 and also with Bell Labs, Murray Hill, NJ 07974.}\makeatother }
\section{Introduction}
 Key-value stores\footnote{Key-value stores  are shared databases that store the data as a collection of key-value pairs.} are essential for many applications such as reservation systems, financial transactions and distributed computing. Such systems  commonly replicate the data across multiple servers to make the data available and accessible with low latency despite the possible failures and  stragglers.  In these systems, the data is frequently updated and it is desirable to make the latest version of the data accessible by the different users.  This requirement is known as consistency in distributed systems \cite{lynch1996distributed}. In order to ensure strong consistency, these systems use strict quorums where the write and the read quorums must intersect \cite{lynch1996distributed,attiya1995sharing}. Specifically, in a system of $N$ servers, where $W$ and $R$ denote the write and the read quorum sizes respectively, $W$ and $R$ are chosen such that $W+R > N$.
 
  In order to have fast access to the data that is critical for many applications, many key-values stores including Amazon's Dynamo \cite{decandia2007dynamo} and Cassandra \cite{lakshman2010cassandra} allow non-strict (partial, probabilistic or sloppy) quorums where $W+R \leq N $. In partial quorum systems  however the write and the read quorums may not intersect. These systems rely on the idea that the write quorum expands as the data propagates to more servers and the write and the read quorum will eventually intersect. Hence, these systems  only guarantee that the users will eventually return the latest version of the data if there are no new write operations\cite{abadi2012consistency, vogels2009eventually}. However, eventual consistency does not specify how soon this will happen. 
 
 Several works studied probabilistic quorum systems, attempted to quantify the staleness of the data retrieved, how soon users can retrieve consistent data and providing adaptive consistency guarantees depending on the application including \cite{malkhi2001probabilistic, wang2010application,sakr2011clouddb,wada2011data,bailis2012probabilistically,chihoub2012harmony,chihoub2013consistency,bailis2014quantifying,golab2014client,liu2015quantitative,mckenzie2015fine,chatterjee2017brief,rahman2017characterizing, zhong2018minimizing}. In \cite{malkhi2001probabilistic}, $\epsilon$-intersecting probabilistic quorum systems were designed such that the probability that any two quorums do not intersect is at most $\epsilon$. In \cite{chihoub2012harmony}, an adaptive approach was proposed that tunes the inconsistency probability, assuming that the response time of the servers are neglected, through controlling the number of servers involved in the read operations at the run-time based on a monitoring module. The monitoring module provides a real-time estimate of the network delays. In this approach, the write operation completes when any server  responds to the write client. While the data is being propagated to the remaining servers, any server is pessimistically considered stale except the first server that responded to the write operation. Hence, this approach does not fully capture expanding write quorums (anti-entropy) \cite{demers1987epidemic}.

In \cite{bailis2012probabilistically, bailis2014quantifying}, the trade-off that partial quorum systems provide between the staleness of the retrieved data and the latency was studied in $3$-way replication-based key-value stores. Specifically, this work answered the question of how stale is the retrieved data through the notion of $l$-staleness, which measures the probability that the users retrieve one of the $l$ latest complete versions. The question of how eventual a user can read consistent data is also studied in \cite{bailis2012probabilistically, bailis2014quantifying} through the notion of $t$-visibility that measures the  probability of returning the value of a write operation $t$ units of times after it completes. While the write operation completes upon receiving acknowledgments from any $W$ servers,  more servers receive the write request after that and the write quorum can continue to expand. Characterizing the $t$-visibility is challenging as it depends on how the write quorum expands based on the delays of the write and the read requests. Hence, the study of \cite{bailis2014quantifying} focused on obtaining insights about this question for $3$-way replication through Monte Carlo simulations. Such simulations however need to be done for all possible values of the quorum sizes and the write  and  read delays. 

Our work instead aims to understand probabilistic quorum system theoretically and to derive closed-form expressions for the inconsistency probability in terms of those parameters. For the widely-used $3$-way replication-based systems, we derive an explicit simple closed-form expression for the inconsistency probability in terms of those parameters. In addition, we provide an approach for analyzing  replication schemes in general with any replication factor.  Finally, our approach can be extended for erasure-coded key-value stores which have not  been investigated before even using  Monte Carlo simulations. 

The rest of this paper is organized as follows. In Section \ref{Model}, we describe the system model. In Section \ref{Quorums}, we study expanding quorums and characterize the probability mass function of the write quorum size. We analyze the inconsistency probability of  partial quorum systems in Section \ref{Consistency Replication}. Finally, concluding remarks are discussed in Section \ref{Conclusion}.

\section{System Model}
\label{Model}
In this section, we describe our system model. We consider a distributed system with $N$ servers denoted by $\mathcal N=\{1, 2, \cdots, N\}$ storing a shared object. A client that issues a write request sends the request to all servers and waits for the acknowledgment of $W$ servers for the write operation to complete. We denote the time that a write request takes to reach to server $i$ in addition to the server's response time by $X_i$, where $i \in \mathcal N$. We assume that $X_1, X_2, \cdots, X_N$ are independent and identically distributed  exponential random variables with parameter $\lambda$. A client that issues a read request sends the request to all servers and waits for $R$ servers to respond. The time the read request takes to reach server $i$ and the server's response time is denoted by $Z_i, i \in \mathcal N$. We assume that the read delays $Z_1, Z_2, \cdots, Z_N$ are independent and identically distributed random variables according to exponential distribution with parameter $\xi$. Finally, we assume that write and read acknowledgments are instantaneous.

In strict quorum systems, $W$ and $R$ are chosen such that $W+R > N$. In partial quorum systems however, $W+R \leq N$ and hence the write and the read quorums may not intersect. This may result in a consistency violation. In real-world quorum systems however, the write quorum expands as the write request propagate to more servers. In \cite{bailis2012probabilistically}, the notion of $t$-visibility was developed  to capture the probability of inconsistency for expanding quorums for a read operation that starts $t$ units of time after the write completes. Our goal in this work is to characterize the inconsistency probability for expanding quorums as a function of $t$, the quorum sizes and the network delays.

\section{Expanding Quorums}
\label{Quorums}
In this section, we characterize the probability distribution of the number of servers in the write quorum $t$ units of time after the write completes. As we have explained, a client that issues a write request sends the request to all $N$ servers and waits to receive acknowledgments from any $W$ servers. The first $W$ received responses determine the write latency $X_{(W)}$, where $X_{(i)}$ denotes the $i$-th smallest of $X_1, X_2, \cdots, X_N$. However, the write quorum will continue to expand as more servers receive the write request. We denote the set of servers that have received the write value $t$ units of time after it completes by $\mathcal S(t)$, where $S(t) \coloneqq |\mathcal S(t)|$ and $S(0)=W$. In Theorem \ref{t-visibility for Exponential Delays}, we characterize the probability mass function (PMF) of $S(t)$.
\begin{theorem}[Dynamic Quorum Size] 
	\label{t-visibility for Exponential Delays}	
	The PMF of the number of servers that have received a complete version $t$ units of time after it completes, $S(t)$, is given by
	\begin{align}
	&\Pr[S(t)=W]=e^{-\lambda_{W+1} t},  \\
	&\Pr[S(t)=s]=\sum\limits_{i=W+1}^{s+1} (-1)^{s-i} (1-e^{- \lambda_i t}) \notag \\ & \  \ \ \ \ \ \ \ \ \ \ \ \ \ \ \ \ \binom{N-W}{N-i+1} \binom{N-i+1}{s-i+1}, 
	\end{align}	
	for $s \in \{W+1, W+2, \cdots, N\}$, where $\lambda_i=(N-i+1) \lambda$.
\end{theorem}
\noindent We provide the proof of Theorem \ref{t-visibility for Exponential Delays} in Appendix A. 

\noindent In Fig. \ref{Prob_Mass_c_W_1_c_R_1}, we show the PMF of $S(1)$ for  $N=3, W=1$ and $\lambda=1$. In Fig. \ref{Prob_Mass_c_W_2_c_R_1}, we show the PMF of $S(1)$ for $N=3, W=2$ and $\lambda=1$.

\begin{figure}[h]
	\centering
	\includegraphics[scale=0.125]{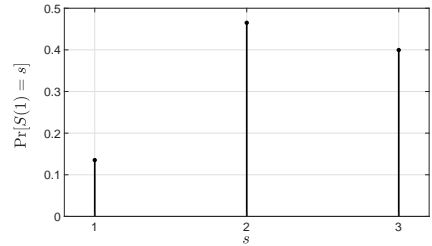}
	\caption{The probability mass function of $S(1)$ for the case where $N=3, W=1$ and $\lambda=1$\label{Prob_Mass_c_W_1_c_R_1}. }
\end{figure}

\begin{figure}[h]
	\centering
	\includegraphics[scale=0.125]{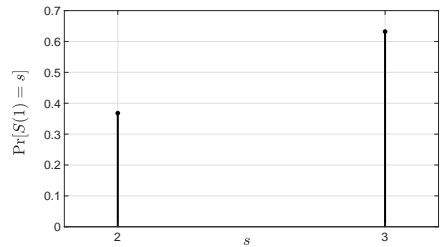}
	\caption{The probability mass function of $S(1)$ for the case where $N=3, W=2$ and $\lambda=1$\label{Prob_Mass_c_W_2_c_R_1}.}
\end{figure}
\section{Consistency Analysis}
\label{Consistency Replication}
In this section, we study the inconsistency probability of  partial quorum systems.
The inconsistency probability is the probability that the write quorum and the read quorum do not intersect. The worst-case probability of inconsistency assuming non-expanding write quorums and instantaneous reads \cite{bailis2012probabilistically} is given by 
\begin{align}
\label{equation: inconsistency probability}
p=\frac{\binom{N-W}{R}} {\binom{N}{R}}.
\end{align}
Since the write quorum expands as the write request propagate to more servers, equation (\ref{equation: inconsistency probability}) is in fact an upper bound of the inconsistency probability \cite{bailis2012probabilistically}. Our objective in this section is to characterize the exact inconsistency probability for expanding quorums. The read client returns inconsistent data if the first $R$ servers that respond to the read request return stale data. A server is considered stale if it replies to the read request before receiving the latest compete version. That is, server $i$ is stale if 
\begin{align}
X_{(W)}+t+Z_i < X_i. 
\end{align}
Denote the first $R$ servers that respond to the read request by $\mathcal R=\{r_1, r_2, \cdots, r_{R}\}$, where $r_1$ is the server the replies first, $r_2$ is the server that replies second and so on. The event that server $r_j$ is stale is expressed as follows
\begin{align}
E_j = \{X_{(W)}+t+Z_{(j)} < X_{r_j}  \}  
&=\{r_j \notin \mathcal S(t+Z_{(j)}) \}.
\end{align}
In order to keep the notation simple, we denote $\mathcal S(t+Z_{(j)})$ by $\mathcal S_j$. The probability that a read returns stale data $t$ units of time after that latest version completes is the probability that all servers in $\mathcal R$ return stale data. Thus, the inconsistency probability can be expressed as follows
\begin{align}
p_t=\Pr[\text{All servers in} \  \mathcal R \ \text{are stale}] =\Pr \left[ \bigcap\limits_{j=1}^{R}  E_j \right].
\end{align}
\noindent The main challenge in characterizing the inconsistency probability is that the events $E_1, E_2, \cdots, E_R$ are dependent. Hence, we express the probability that these events intersect as follows
\begin{align}
p_t&=\Pr \left[r_1 \notin \mathcal S_1, r_2 \notin \mathcal S_2, \cdots, r_{R} \notin \mathcal S_{R} \right] \notag \\
&=\Pr \left[r_{R} \notin \mathcal S_{R}| r_{R-1} \notin \mathcal S_{R-1}, \cdots,  r_1 \notin \mathcal S_1 \right] \cdots  \notag \\
& \ \ \ \Pr \left[r_2 \notin \mathcal S_2 |r_1 \notin \mathcal S_1  \right] \Pr \left[r_1 \notin \mathcal S_1 \right].
\end{align}
This dependency requires careful handling. In order to see this, we consider the case where $R=2$ for instance. The probability of inconsistency can be first expressed as follows
\begin{align}
p_t &=\Pr \left[r_2 \notin \mathcal S_2, r_1 \notin \mathcal S_1 \right] \notag \\
&=\Pr \left[r_2 \notin \mathcal S_2 | r_1 \notin \mathcal S_1 \right] \Pr \left[r_1 \notin \mathcal S_1 \right]. 
\end{align} 
 However,  we cannot find the probability of the event $\{r_2 \notin \mathcal S_2 | r_1 \notin \mathcal S_1\}$ directly.  This is because when $r_1 \notin \mathcal S_1$, we need to consider two sub-cases. The first case is the case where $r_1 \notin \mathcal S_2$. The second case is the case where $r_1 \in \mathcal S_2$, hence $r_1 \in \mathcal S_2-\mathcal S_1$. Based on this reasoning,  the inconsistency probability can be expressed as follows 
\begin{align}
p_t &=\Pr \left[r_2 \notin \mathcal S_2 | r_1 \notin \mathcal S_2 \right] \notag \Pr \left[r_1 \notin \mathcal S_2 \right]   \\ &+  \Pr \left[ r_2 \notin \mathcal S_2 | r_1 \in \mathcal S_2-\mathcal S_1\right]  \Pr[r_1 \in \mathcal S_2- \mathcal S_1].
\end{align} 
Each of these terms can be then expressed in terms of the expected network delays  and the quorum sizes to get a closed-form expression for the inconsistency probability in this case.  In Theorem \ref{theorem: consistency violation for replication}, that we prove in Appendix B, we provide our main result in which we characterize the inconsistency probability of the widely-used $3$-way replication technique.
\begin{theorem}[Inconsistency Probability] \  
\label{theorem: consistency violation for replication}	
\begin{itemize}
\item The  worst-case inconsistency probability for the case where $W=1$ and $R=1$ is expressed as follows
\begin{align}
p_t=  \frac{2\xi e^{-\lambda t}}{\lambda+3\xi}.
\end{align}	
\item The  worst-case inconsistency probability for the case where $W=2$ and $R=1$ is expressed as follows
\begin{align}
p_t=  \frac{\xi e^{-\lambda t}}{\lambda+3\xi}.
\end{align}	
\item The worst-case inconsistency probability for the case where $W=1$ and $R=2$ is expressed as follows
\begin{align}
p_t &=\frac{6\xi^3e^{-2\lambda t}}{(\lambda+2\xi)(\lambda+3\xi)} \notag \\ &\left(\frac{2 \lambda}{(\lambda+2\xi)(\lambda+3 \xi)}-  \frac{ (\lambda-\xi) e^{-\lambda t}}{(\lambda+\xi)(2\lambda+3 \xi)}\right).
\end{align}
\end{itemize}
\end{theorem}
\begin{remark}
It can be verified that at $t=0$, the limit of the inconsistency probability of Theorem \ref{theorem: consistency violation for replication} as $\xi$ grows is equal to the inconsistency probability assuming instantaneous reads given in (\ref{equation: inconsistency probability}). That is, we have
\begin{align}
 \lim_{\xi \to\infty} p_0=p.
 \end{align}
\end{remark}
\noindent It is worth noting that the upper bound of the inconsistency probability given in (\ref{equation: inconsistency probability}) is quite loose. In order to see this, we observe that this bound gives an inconsistency probability of $1/3$ for the case where $W=2, R=1$ and also for the case where $W=1, R=2$. Hence, this bound does not differentiate between these two cases. We show the probability of inconsistency for the different cases in Fig. \ref{Inconsistency_Prob_Replication}.
\begin{figure}[h]
	\centering
	\includegraphics[scale=0.185]{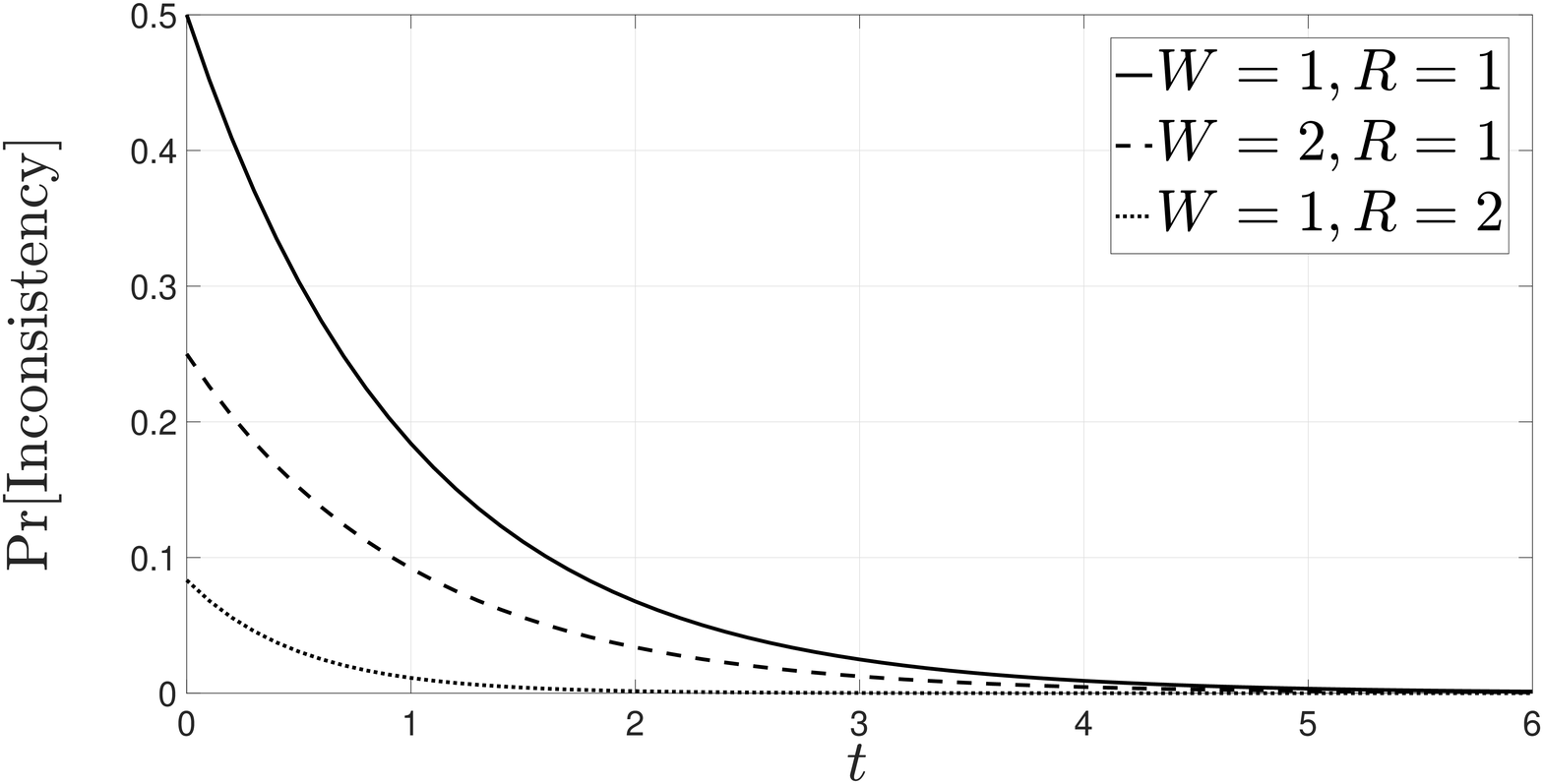}
	\caption{The probability of inconsistency for the case where $N=3, \lambda=1$ and $\xi=1$\label{Inconsistency_Prob_Replication}.}
\end{figure}

\begin{remark}[Asymmetry]
It is worth noting that the inconsistency probability is asymmetric in the write and read quorum sizes and also the write and read mean delays.   
\end{remark}
\begin{figure}[h]
	\centering
	\includegraphics[scale=0.45]{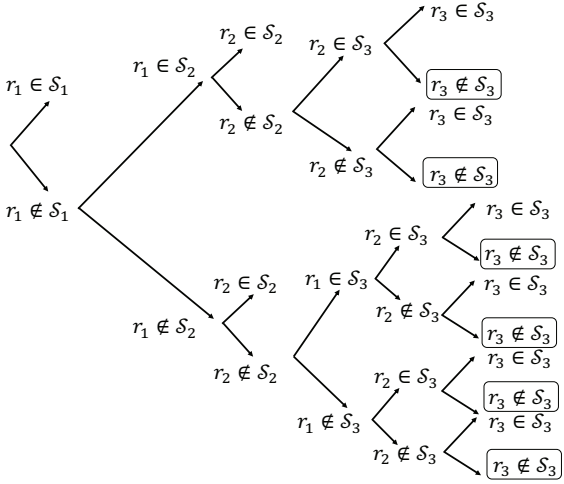}
	\caption{Inconsistency cases for the case where $R = 3$\label{Inconsistency_cases_R_3}.}
\end{figure}
 \begin{remark}[Replication Factor]
 While the case of $N=3$ is the typical case in replication-based systems, we can derive the inconsistency probability for general $N, W$ and $R$ \cite[Ch.\ 4]{alithesis}. There are $R !$ cases to be considered. For instance, for $R=3$, the following cases shown in Fig. \ref{Inconsistency_cases_R_3} lead to violating the consistency
	\begin{enumerate}
	\item $\left(r_3 \notin \mathcal S_3, r_2 \in \mathcal S_3-\mathcal S_2, r_1 \in \mathcal S_2-\mathcal S_1 \right)$,
	\item  $\left(r_3 \notin \mathcal S_3, r_2 \notin \mathcal S_3, r_1 \in \mathcal S_2-\mathcal S_1\right)$, 
	\item $\left(r_3 \notin \mathcal S_3, r_2 \in \mathcal S_3-\mathcal S_2, r_1 \in \mathcal S_3-\mathcal S_2 \right)$,
	\item $\left(r_3 \notin \mathcal S_3, r_2 \notin \mathcal S_3, r_1 \in \mathcal S_3-\mathcal S_2 \right)$, 
	\item $\left(r_3 \notin \mathcal S_3, r_2 \in \mathcal S_3-\mathcal S_2, r_1 \notin \mathcal S_3 \right)$,
	\item  $\left(r_3 \notin \mathcal S_3, r_2 \notin \mathcal S_3, r_1 \notin \mathcal S_3 \right)$. 
\end{enumerate}

 \end{remark}
\begin{remark} [Beyond Replication and Exponential Delays]
The proof technique of Theorem \ref{theorem: consistency violation for replication} can be used to characterize the inconsistency probability for any given distributions of the write and read delays such as shifted exponential distribution. This approach can be extended also to simple erasure-coded probabilistic key-value stores, where each version of the data is encoded using a maximum distance separable (MDS) code of dimension $k$. In erasure-coded strict quorum systems, $W$ and $R$ are chosen such that $W+R-N \geq k$ to ensure that the write and read quorums intersect in at least $k$ servers. In erasure-coded probabilistic quorum systems however, the quorum sizes may be selected such that $W+R-N < k$ to provide faster access to the data. Characterizing the inconsistency probability of these systems is challenging and we refer the reader to \cite[Ch.\ 5]{alithesis} for a follow up in this direction.  
\end{remark}

\section{Conclusion}  
\label{Conclusion}
In this paper, we have studied the consistency-latency trade-off for Dynamo-style replication-based key-value stores analytically and  derived a closed-form expression for the inconsistency probability for the $3$-way replication technique. Our study allows fine-tuning of latency and consistency guarantees based on the mean values of the write and read delays of the data store. An immediate future work is to incorporate our tuning policy in a distributed key-value store and evaluate its performance. Extending this study to derive a tight upper bound on the inconsistency probability for any given distributions of the write delays, read delays and acknowledgments delays is also an interesting future research direction.
\section*{Acknowledgment}
The author would like to thank Viveck Cadambe and Mohammad Fahim for their helpful comments.

\bibliographystyle
{IEEEtran}
\bibliography{IEEEabrv,Nulls}

\begin{thebibliography}{10}
\providecommand{\url}[1]{#1}
\csname url@samestyle\endcsname
\providecommand{\newblock}{\relax}
\providecommand{\bibinfo}[2]{#2}
\providecommand{\BIBentrySTDinterwordspacing}{\spaceskip=0pt\relax}
\providecommand{\BIBentryALTinterwordstretchfactor}{4}
\providecommand{\BIBentryALTinterwordspacing}{\spaceskip=\fontdimen2\font plus
\BIBentryALTinterwordstretchfactor\fontdimen3\font minus
  \fontdimen4\font\relax}
\providecommand{\BIBforeignlanguage}[2]{{%
\expandafter\ifx\csname l@#1\endcsname\relax
\typeout{** WARNING: IEEEtran.bst: No hyphenation pattern has been}%
\typeout{** loaded for the language `#1'. Using the pattern for}%
\typeout{** the default language instead.}%
\else
\language=\csname l@#1\endcsname
\fi
#2}}
\providecommand{\BIBdecl}{\relax}
\BIBdecl

\bibitem{lynch1996distributed}
N.~A. Lynch, \emph{Distributed algorithms}.\hskip 1em plus 0.5em minus
  0.4em\relax Elsevier, 1996.

\bibitem{attiya1995sharing}
H.~Attiya, A.~Bar-Noy, and D.~Dolev, ``Sharing memory robustly in
  message-passing systems,'' \emph{Journal of the ACM (JACM)}, vol.~42, no.~1,
  pp. 124--142, 1995.

\bibitem{decandia2007dynamo}
G.~DeCandia, D.~Hastorun, M.~Jampani, G.~Kakulapati, A.~Lakshman, A.~Pilchin,
  S.~Sivasubramanian, P.~Vosshall, and W.~Vogels, ``Dynamo: amazon's highly
  available key-value store,'' in \emph{ACM SIGOPS operating systems review},
  vol.~41, no.~6.\hskip 1em plus 0.5em minus 0.4em\relax ACM, 2007, pp.
  205--220.

\bibitem{lakshman2010cassandra}
A.~Lakshman and P.~Malik, ``Cassandra: a decentralized structured storage
  system,'' \emph{ACM SIGOPS Operating Systems Review}, vol.~44, no.~2, pp.
  35--40, 2010.

\bibitem{abadi2012consistency}
D.~Abadi, ``Consistency tradeoffs in modern distributed database system design:
  Cap is only part of the story,'' \emph{Computer}, vol.~45, no.~2, pp. 37--42,
  2012.

\bibitem{vogels2009eventually}
W.~Vogels, ``Eventually consistent,'' \emph{Communications of the ACM},
  vol.~52, no.~1, pp. 40--44, 2009.

\bibitem{malkhi2001probabilistic}
D.~Malkhi, M.~K. Reiter, A.~Wool, and R.~N. Wright, ``Probabilistic quorum
  systems,'' \emph{Information and Computation}, vol. 170, no.~2, pp. 184--206,
  2001.

\bibitem{wang2010application}
X.~Wang, S.~Yang, S.~Wang, X.~Niu, and J.~Xu, ``An application-based adaptive
  replica consistency for cloud storage,'' in \emph{2010 Ninth International
  Conference on Grid and Cloud Computing}.\hskip 1em plus 0.5em minus
  0.4em\relax IEEE, 2010, pp. 13--17.

\bibitem{sakr2011clouddb}
S.~Sakr, L.~Zhao, H.~Wada, and A.~Liu, ``Clouddb autoadmin: Towards a truly
  elastic cloud-based data store,'' in \emph{2011 IEEE International Conference
  on Web Services}, 2011, pp. 732--733.

\bibitem{wada2011data}
H.~Wada, A.~Fekete, L.~Zhao, K.~Lee, and A.~Liu, ``Data consistency properties
  and the trade-offs in commercial cloud storage: the consumers' perspective.''
  in \emph{CIDR}, vol.~11, 2011, pp. 134--143.

\bibitem{bailis2012probabilistically}
P.~Bailis, S.~Venkataraman, M.~J. Franklin, J.~M. Hellerstein, and I.~Stoica,
  ``Probabilistically bounded staleness for practical partial quorums,''
  \emph{Proceedings of the VLDB Endowment}, vol.~5, no.~8, pp. 776--787, 2012.

\bibitem{chihoub2012harmony}
H.-E. Chihoub, S.~Ibrahim, G.~Antoniu, and M.~S. Perez, ``Harmony: Towards
  automated self-adaptive consistency in cloud storage,'' in \emph{2012 IEEE
  International Conference on Cluster Computing}, pp. 293--301.

\bibitem{chihoub2013consistency}
------, ``Consistency in the cloud: When money does matter!'' in \emph{2013
  13th IEEE/ACM International Symposium on Cluster, Cloud, and Grid Computing},
  2013, pp. 352--359.

\bibitem{bailis2014quantifying}
P.~Bailis, S.~Venkataraman, M.~J. Franklin, J.~M. Hellerstein, and I.~Stoica,
  ``Quantifying eventual consistency with {P}{B}{S},'' \emph{The VLDB Journal},
  vol.~23, no.~2, pp. 279--302, 2014.

\bibitem{golab2014client}
W.~Golab, M.~R. Rahman, A.~AuYoung, K.~Keeton, and I.~Gupta, ``Client-centric
  benchmarking of eventual consistency for cloud storage systems,'' in
  \emph{2014 IEEE 34th International Conference on Distributed Computing
  Systems}, 2014, pp. 493--502.

\bibitem{liu2015quantitative}
S.~Liu, S.~Nguyen, J.~Ganhotra, M.~R. Rahman, I.~Gupta, and J.~Meseguer,
  ``Quantitative analysis of consistency in nosql key-value stores,'' in
  \emph{International Conference on Quantitative Evaluation of Systems}.\hskip
  1em plus 0.5em minus 0.4em\relax Springer, 2015, pp. 228--243.

\bibitem{mckenzie2015fine}
M.~McKenzie, H.~Fan, and W.~Golab, ``Fine-tuning the consistency-latency
  trade-off in quorum-replicated distributed storage systems,'' in \emph{2015
  IEEE International Conference on Big Data}, pp. 1708--1717.

\bibitem{chatterjee2017brief}
S.~Chatterjee and W.~Golab, ``Brief announcement: A probabilistic performance
  model and tuning framework for eventually consistent distributed storage
  systems,'' in \emph{Proceedings of the ACM Symposium on Principles of
  Distributed Computing}.\hskip 1em plus 0.5em minus 0.4em\relax ACM, 2017, pp.
  259--261.

\bibitem{rahman2017characterizing}
M.~R. Rahman, L.~Tseng, S.~Nguyen, I.~Gupta, and N.~Vaidya, ``Characterizing
  and adapting the consistency-latency tradeoff in distributed key-value
  stores,'' \emph{ACM Transactions on Autonomous and Adaptive Systems (TAAS)},
  vol.~11, no.~4, p.~20, 2017.

\bibitem{zhong2018minimizing}
J.~Zhong, R.~D. Yates, and E.~Soljanin, ``Minimizing content staleness in
  dynamo-style replicated storage systems,'' in \emph{IEEE INFOCOM AoI
  Workshop}, 2018, pp. 361--366.

\bibitem{demers1987epidemic}
A.~Demers, D.~Greene, C.~Hauser, W.~Irish, J.~Larson, S.~Shenker, H.~Sturgis,
  D.~Swinehart, and D.~Terry, ``Epidemic algorithms for replicated database
  maintenance,'' in \emph{Proceedings of the sixth annual ACM Symposium on
  Principles of distributed computing}, 1987.

\bibitem{alithesis}
R.~E. Ali, \emph{Harnessing Data Correlation and Network Information in
  Distributed Key-Value Stores}.\hskip 1em plus 0.5em minus 0.4em\relax PhD
  Dissertation, Penn State University, 2020.

\bibitem{renyi1953theory}
A.~R{\'e}nyi, ``On the theory of order statistics,'' \emph{Acta Mathematica
  Hungarica}, vol.~4, no. 3-4, pp. 191--231, 1953.

\bibitem{bibinger2013notes}
M.~Bibinger, ``Notes on the sum and maximum of independent exponentially
  distributed random variables with different scale parameters,'' \emph{arXiv
  preprint arXiv:1307.3945}, 2013.

\end{thebibliography}
\section{Appendices}
\label{Appendices}
We begin with a brief background about exponential random variables that we build on later to prove our results. We first recall the following useful Lemma \cite{renyi1953theory} for the order statistics of independent exponential random variables with a common parameter $\lambda$.
\begin{lemma}[Order Statistics of Independent Exponentials] 
	\label{Order Statistics}	
	Let $X_1, X_2, \cdots, X_n$ be independent and identically distributed random variables according to $\rm exp(\lambda)$, then we have
	\begin{align}
	Y_i \coloneqq X_{(i)}-X_{(i-1)} \sim \mathrm{exp}((n-i+1) \lambda),
	\end{align} 
	where $X_{(i)}$ denotes the $i$-th smallest of $X_1, X_2, \cdots, X_n$, $i \in \{1, 2 ,\cdots, n\}$ and $X_{(0)}=0$.
\end{lemma}

\noindent We also recall the following Lemma from \cite{bibinger2013notes} which studies the sum of independent exponential random variables with different parameters.
\begin{lemma}[Sum of Exponentials]
	\label{Sum of Exponentials}	
	Let $Y_1, Y_2, \cdots, Y_n$ be independent exponentials random variables with parameters $\lambda_1, \lambda_2, \cdots, \lambda_n$ respectively, where $f_{Y_i}(y)$ denotes the density function of $Y_i$. The density function of  
	\begin{align}
	Z\coloneqq \sum\limits_{i=1}^{n} Y_i
	\end{align}
	is given by
	\begin{align}
	f_Z(z)= \sum\limits_{i=1}^{n} f_i(z) \prod_{\substack{j=1, \\ j \neq i }}^{n} \frac{\lambda_j}{\lambda_j-\lambda_i} , \ z \geq 0.
	\end{align}	  
\end{lemma}

\subsection{Proof of Theorem \ref{t-visibility for Exponential Delays}}
\noindent We   are now  ready to prove Theorem \ref{t-visibility for Exponential Delays}. 
\begin{proof}

	For the case where $s=W$, we have	
	\begin{align*}	
	\Pr[S(t)=W]&=\Pr[S(t) \leq W] \notag \\
	&=\Pr[ X_{(W+1)}- X_{(W)} > t]\notag \\
 &=e^{-\lambda_{W+1} t}, 
	\end{align*}
	where the last equality follows Lemma \ref{Order Statistics}.\\ 
	For the case were $s \in \{W+1, W+2, \cdots, N\}$, we have
	\begin{align*}
	\Pr[&S(t)=s]=\Pr[S(t) \leq s]-\Pr[S(t) \leq s-1]  \\
	&=\Pr[ X_{(s+1)}- X_{(W)} > t] -\Pr[ X_{(s)}- X_{(W)} > t]   \\
	&=\Pr[ X_{(s)}- X_{(W)} \leq  t]-\Pr[ X_{(s+1)}- X_{(W)} \leq t] \notag  \\
	&=\Pr \left[\sum\limits_{i=W+1}^{s}  X_{(i)} - X_{(i-1)} \leq t \right] \\
	&-\Pr \left[\sum\limits_{i=W+1}^{s+1}  X_{(i)}- X_{(i-1)} \leq t \right]  \notag \\
	&=\Pr \left[\sum\limits_{i=W+1}^{s}  Y_i \leq t \right] -\Pr \left[\sum\limits_{i=W+1}^{s+1}  Y_i\leq t \right],
	\end{align*}
	where $Y_i = X_{(i)}-X_{(i-1)}$. 
	
\noindent Since $ X_1,  X_2, \cdots,  X_N$ are independent and identical  exponential random variables, then $Y_i$ is an exponential random variable with parameter $\lambda_i=(N-i+1) \lambda$, where $i \in \{2, 3, \cdots, N\}$ from Lemma \ref{Order Statistics}. Since $ Y_{1},  Y_{2}, \cdots,  Y_{N}$ are independent exponential random variables, from Lemma \ref{Sum of Exponentials}, we have
\begin{small}
	\begin{align*}
	 &\Pr[S(t)=s] =\Pr \left[\sum\limits_{i=W+1}^{s}  Y_i \leq t \right] -\Pr \left[\sum\limits_{i=W+1}^{s+1}  Y_i\leq t \right] \\ &= \sum\limits_{i=W+1}^{s} F_i(t) \prod_{\substack{j=W+1, \\ j \neq i }}^{s} \frac{\lambda_j}{\lambda_j-\lambda_i} -\sum\limits_{i=W+1}^{s+1} F_i(t) \prod_{\substack{j=W+1, \\ j \neq i }}^{s+1} \frac{\lambda_j}{\lambda_j-\lambda_i}  \\
	&=\sum\limits_{i=W+1}^{s} F_i(t) \frac{\lambda_i}{\lambda_i-\lambda_{s+1}} \prod_{\substack{j=W+1, \\ j \neq i }}^{s} \frac{\lambda_j}{\lambda_j-\lambda_i} 
	- F_{s+1}(t)\prod_{j=W+1}^{s} \frac{\lambda_j}{\lambda_j-\lambda_{s+1}}  \\ 
	&=\sum\limits_{i=W+1}^{s} F_i(t) \frac{N-i+1}{s-i+1} \prod_{\substack{j=W+1, \\ j \neq i }}^{s} \frac{N-j+1}{i-j} 
	-F_{s+1}(t)\prod_{j=W}^{s-1} \frac{N-j}{s-j} \\
	&=\sum\limits_{i=W+1}^{s} F_i(t) \frac{N-i+1}{s-i+1} \prod_{\substack{j=W+1, \\ j \neq i }}^{s} \frac{N-j+1}{i-j} 
	-F_{s+1}(t) \binom{N-W}{N-s} \\
	&=\sum\limits_{i=W+1}^{s} (1-e^{- \lambda_i t}) \frac{N-i+1}{s-i+1} \prod_{\substack{j=W+1, \\ j \neq i }}^{s} \frac{N-j+1}{i-j} \\
	&-(1-e^{-\lambda_{s+1}t}) \binom{N-W}{N-s} \\
	&=\sum\limits_{i=W+1}^{s+1} (-1)^{s-i} (1-e^{- \lambda_i t}) \binom{N-W}{N-i+1} \binom{N-i+1}{s-i+1}.
	\end{align*}
\end{small}	
\end{proof}

\subsection{Proof of Theorem \ref{theorem: consistency violation for replication}}
\noindent  In order to find the inconsistency probability,  we first need to characterize the PMF of the number of servers in the write quorum $t+Z_{(j)}$ units of time after the write completes as  given in Lemma \ref{Lemma: pmf}.
\begin{lemma}
	\label{Lemma: pmf}	
	The PMF of the number of servers in the write quorum $t+Z_{(j)}$ units of time, where $j \in \mathcal R$, after the write completes is given by
	\begin{align}
	& \Pr[S(t+Z_{(j)})=W] =  e^{-\lambda_{W+1}t} \\
	& ~~~~~~~~~~~~~~~~~~~~~~~\sum_{l=1}^{j}  \binom{N}{j} \binom{j}{l} \frac{(-1)^{j-l} \xi_{N-l+1}}{\xi_l+\lambda_{W+1}},  \notag \\	
	& \Pr[S(t+Z_{(j)})=s]= \notag \\ &\sum\limits_{i=W+1}^{s+1} (-1)^{s-i}   \binom{N-W}{N-i+1} \binom{N-i+1}{s-i+1} 
	\notag \\ &\left(1-e^{-\lambda_i t} \sum_{l=1}^{j}  \binom{N}{j} \binom{j}{l} \frac{(-1)^{j-l} \xi_{N-l+1}}{\xi_l+\lambda_i}\right),
	\end{align}	
	for $s \in \{W+1, \cdots, N\}$, where $\xi_j=(N-j+1) \xi$ and $\lambda_j=(N-j+1) \lambda$.
\end{lemma}
\begin{proof}
	Based on Lemma \ref{Sum of Exponentials}, we can express the probability density function of  $Z_{(j)}=\sum_{l=1}^{j} Z_{(l)}-Z_{(l-1)}$ as follows
	\begin{align*}
	f_{Z_{(j)}}(z) &= \sum_{l=1}^{j} f_l(z) \prod_{\substack{i=1, \\ i \neq l}}^{j} \frac{\xi_i}{\xi_i-\xi_l} 
	\\ &= \sum_{l=1}^{j} (-1)^{j-l} \xi_{N-l+1} \binom{N}{j} \binom{j}{l}  e^{-\xi_l z},
	\end{align*}	
	where $z \geq 0$. Therefore, from Theorem \ref{t-visibility for Exponential Delays}, we can express $\Pr[S(t+Z_{(j)})=W]$ as follows 
	\begin{align*}
	\Pr[S(t+Z_{(j)})&=W] = \int_{0}^{\infty} e^{-\lambda_{W+1}(t+z)}  f_{Z(j)}(z) \ dz 
	\\ &= e^{-\lambda_{W+1}t}  \sum_{l=1}^{j}  \binom{N}{j} \binom{j}{l} \frac{(-1)^{j-l} \xi_{N-l+1}}{\xi_l+\lambda_{W+1}},
	\end{align*}
	where $\xi_j=(N-j+1) \xi$ and $\lambda_j=(N-j+1) \lambda$. Similarly for $s \in \{W+1, W+2, \cdots, N\}$, we have
	\begin{small}
	\begin{align*}
	&\Pr[S(t+Z_{(j)})=s]=\sum\limits_{i=W+1}^{s+1} (-1)^{s-i}  \binom{N-W}{N-i+1} \binom{N-i+1}{s-i+1} 
	\\ &~~~~~~~~~~~~~~~~~~~~~~~~~~\left(1-e^{-\lambda_i t} \sum_{l=1}^{j}  \binom{N}{j} \binom{j}{l} \frac{(-1)^{j-l} \xi_{N-l+1}}{\xi_l+\lambda_i}\right).
	\end{align*}	
	\end{small}
\end{proof}

\noindent We are now ready to prove Theorem \ref{theorem: consistency violation for replication}.
\begin{proof}
 The probability of inconsistency for the case where $W=1$ and $R=1$ can be expressed as follows
	\begin{align*}
	p_t&=\Pr[r_1 \notin \mathcal S_1] \\&=\sum_{s=W}^{N}\Pr[r_1 \notin \mathcal S_1 | S(t+Z_{(1)})=s] \\ & \ \ \ \Pr[S(t+Z_{(1)})=s] \\
	&=\sum_{s=W}^{N} \left(1-\frac{s}{N} \right) \Pr[S(t+Z_{(1)})=s]  \\
	&= \frac{2}{3} \Pr[S(t+Z_{(1)})=1]+ \frac{1}{3} \Pr[S(t+Z_{(1)})=2]  \\
	&= \frac{2}{3}\frac{ \xi_1 e^{-2\lambda t}}{\xi_1+2\lambda}+ \frac{1}{3} \left( \frac{2 \xi_1 e^{-\lambda t}}{\xi_1+\lambda} -\frac{2 \xi e^{-2 \lambda t}}{\xi_1+2 \lambda}\right)  \\
	&=\frac{2}{3} \frac{\xi_1 e^{-\lambda t}}{\xi_1+\lambda} 
	=\frac{2 \xi e^{-\lambda t}}{3\xi+\lambda}.
	\end{align*}
	Similarly, for the case where $W=2$ and $R=1$, we have
	\begin{align*}
	p_t= \frac{1}{3} \Pr[S(t+Z_{(1)})=2]= \frac{1}{3} \frac{\xi_1 e^{-\lambda t}}{\xi_1+\lambda}=\frac{\xi e^{-\lambda t}}{3\xi+\lambda}.
	\end{align*}
	For the case where $R=2$, we can express the probability of inconsistency as follows
		\begin{align*}
	p_t &=\Pr \left[r_2 \notin \mathcal S_2, r_1 \notin \mathcal S_1 \right] \notag \\
	&=\Pr \left[r_2 \notin \mathcal S_2 | r_1 \notin \mathcal S_1 \right] \Pr \left[r_1 \notin \mathcal S_1 \right]. 
	\end{align*} 
    If $r_1 \notin \mathcal S_1$, it may happen that $r_1 \notin \mathcal S_2$ as well or $r_1 \in \mathcal S_2$ and these two cases need to be handled separately. Therefore, we express the inconsistency probability as follows 
	\begin{align*}
	&p_t =\Pr \left[r_2 \notin \mathcal S_2, r_1 \notin \mathcal S_1 \right] \notag \\
	&= \Pr \left[r_2 \notin \mathcal S_2 | r_1 \notin \mathcal S_2 \right] \notag \Pr \left[r_1 \notin \mathcal S_2 \right]   \\ &+  \Pr \left[ r_2 \notin \mathcal S_2 | r_1 \in \mathcal S_2-\mathcal S_1\right]  \Pr[r_1 \in \mathcal S_2- \mathcal S_1].
	\end{align*} 
It is important to note that 
	\begin{align*}
	 \Pr \left[ r_2 \notin \mathcal S_2 | r_1 \in \mathcal S_2-\mathcal S_1\right]= \Pr \left[ r_2 \notin \mathcal S_2 | r_1 \in \mathcal S_2 \right].  
	\end{align*} 
   Hence, we have
\begin{align*}
&p_t = \Pr \left[r_2 \notin \mathcal S_2 | r_1 \notin \mathcal S_2 \right] \notag \Pr \left[r_1 \notin \mathcal S_2 \right]   \\ &+  \Pr \left[ r_2 \notin \mathcal S_2 | r_1 \in \mathcal S_2 \right]  ( \Pr[r_1 \in \mathcal S_2]-\Pr[r_1 \in \mathcal S_1]),
\end{align*} 
where
	\begin{align}
	\Pr \left[r_2 \notin \mathcal S_2 | r_1 \notin \mathcal S_2 \right]&=  \sum_{s=W}^{N-1} \left(1-\frac{s}{N-1}\right) \notag \\ 
	& \ \ \ \ \Pr \left[ S(t+Z_{(2)})=s\right] \notag  \\
	&= \frac{1}{2} \Pr \left[ S(t+Z_{(2)})=1\right],
	\end{align}
	\begin{align}
	\Pr \left[r_2 \notin \mathcal S_2 | r_1 \in \mathcal S_2 \right] &=  \sum_{s=W}^{N} \left(1-\frac{s-1}{N-1}\right) \notag \\ 
	& \ \ \ \ \Pr \left[ S(t+Z_{(2)})=s\right] \notag \\
	&=\frac{1}{2} \Pr \left[ S(t+Z_{(2)})=2\right],
	\end{align}
	\begin{align}
	&\Pr[r_1 \in \mathcal S_1]=\sum_{s=W}^{N} \frac{s}{N} \Pr[S(t+Z_{(1)})=s] \notag \\
	&=\frac{1}{3} \Pr[S(t+Z_{(1)})=1]+ \notag \\ 
	&\frac{2}{3} \Pr[S(t+Z_{(1)})=2] +\Pr[S(t+Z_{(1)})=3],
	\end{align}
	and 
	\begin{align}
&\Pr[r_1 \in \mathcal S_2]=\sum_{s=W}^{N} \frac{s}{N} \Pr[S(t+Z_{(2)})=s] \notag \\
	&=\frac{1}{3} \Pr[S(t+Z_{(2)})=1] \notag \\
	&+\frac{2}{3} \Pr[S(t+Z_{(2)})=2]+ \Pr[S(t+Z_{(2)})=3].
	\end{align}
	Therefore, we can express the probability of inconsistency in this case as follows 
	\small
	\begin{align*}
	p_t&=\frac{6\xi^3e^{-2\lambda t}}{(\lambda+2\xi)(\lambda+3\xi)}. \left(\frac{2 \lambda}{(\lambda+2\xi)(\lambda+3 \xi)}-\frac{ (\lambda-\xi) e^{-\lambda t}}{(\lambda+\xi)(2\lambda+3 \xi)}\right).
	\end{align*}
	\end{proof}
\end{document}